\documentclass[conference]{IEEEtran}
\usepackage{enumerate}
\usepackage{tikz}
\usepackage{epstopdf}
\usetikzlibrary{arrows}
\usepackage{cite}
\usepackage{amsmath,amssymb,amsthm}
\usepackage{latexsym}
\usepackage{array}
\usepackage{tabularx}
\usepackage{epsfig}
\usepackage{bm}
\newtheorem{definition}{Definition}

\newtheorem{proposition}{Proposition}

\newcommand{\nozero}{\backslash \left\{0\right\}}
\newcommand{\diag}[1]{\text{diag}\left( #1 \right)}

\newcommand{\vol}{\mbox{vol}}
\newcommand{\xx}{\mathbf{x}}
\newcommand{\yy}{\mathbf{y}}

\newcommand{\ww}{\mathbf{w}}
\newcommand{\MM}{\mathbf{M}}
\newcommand{\HH}{\mathbf{H}}
\newcommand{\EE}{\mathbf{E}}
\newcommand{\DD}{\mathbf{D}}
\newcommand{\II}{\mathbf{I}}

\newtheorem{lemma}{Lemma}

\begin{document}
	
	\title{Multilevel Code Construction for \\ Compound Fading Channels}
	
	\author{\IEEEauthorblockN{Antonio Campello, Ling Liu and Cong Ling} \IEEEauthorblockA{Department of Electrical and Electronic Engineering\\
			Imperial College London, U.K.\\
			$[$a.campello,l.liu12,c.ling$]$@imperial.ac.uk}
		}
	
	\maketitle
	\begin{abstract}
 We consider explicit constructions of multi-level lattice codes that universally approach the capacity of the compound block-fading channel. Specifically, building on algebraic partitions of lattices, we show how to construct codes with negligible probability of error for any channel realization and normalized log-density approaching the Poltyrev limit. Capacity analyses and numerical results on the achievable rates for each partition level are provided. The proposed codes have several enjoyable properties such as constructiveness and good decoding complexity, as compared to random one-level codes. Numerical results for finite-dimensional multi-level lattices based on polar codes are exhibited.
	\end{abstract}
	\section{Introduction}
	Compound channels are suitable models for open-loop communication scenarios, when a transmitter has limited knowledge of the channel state, but requires reliable transmission for all states. Another natural application is broadcasting a message to multiple receivers, where each transmission link is possibly in a different channel state. A code for a compound channel is \textit{universal}, in the sense that the probability of error vanishes for any channel in a set, known as the \textit{compound set}. The objective of this work is to study practical multi-level constructions of universal codes for block-fading channels. 
	
	It was previously shown in \cite{CLB16} that random algebraic lattice codes achieve the capacity of the compound block-fading (and MIMO) channels. However, as expected, random codes lack encoding and decoding efficiency. To deal with this problem, we employ in this paper the concept of \textit{multilevel codes}, or multi-level lattices, as introduced by Forney et al \cite{Forney00}. Multi-level lattices can be obtained by concatenating a low dimensional chain of lattices with efficient nested codes. In particular, we show that chains of algebraic lattices (e.g. \cite{GBB97})  may be used for practical construction of efficient universal codes, with essentially the same complexity as codes for the Gaussian channel (sections \ref{sec:twoLevel} and \ref{sec:consD}). Leveraging from the algebraic structure, we can identify the best and worst-channels in the compound set. As an illustration, we analyze the behavior of multi-level lattices based on \textit{polar codes} tailored for two extreme channels in each level, as well as a simple scheme based on a surrogate BEC channel in each level (Section \ref{sec:numerical}). Both schemes operate at within a fraction of dB from the limit with small probability of error for all channels and multi-stage decoding.
	
	\subsection*{Related Results}
	
	Traditional random $\mathbb{Z}$-lattice codes as in \cite{ZamirBook} can indeed achieve the compound capacity of the power constrained block-fading channels, with no extra algebraic structure. This is a consequence of a theorem in \cite{HeshamElGamal04} along with a compactness argument (see also \cite[Appx. A]{ShiWesel07}, \cite[Thm. 2]{CLB16} for more details). The integer-forcing architecture \cite{Ordentlich15} provides an effective solution for finding codes that operate at a gap to capacity in the broader class of compound MIMO channels by decoupling equalization and decoding. Algebraic random lattices offers an alternative in the fading case with several advantages, such as guaranteed full-diversity, better quantization of the channel space, while also admitting a decoupled implementation. 

	With respect to practical multi-level constructions, the advantages of algebraic lattices are even more evident. For instance, for integer partitions the achievable rates in each level vary significantly with the channel state (see, e.g., Fig. \ref{fig:OK2OK}), implying a very poor compound capacity. Algebraic partitions, on the other hand, ``absorb'' part of channel realizations, making the capacity of each level vary smoothly enough in the channel space.
	
	Regarding previous algebraic constructions, the work \cite{BB16} studied one-level chains of lattices with full diversity and derived bounds on the alphabet size of the underlying codes. For independent fading channels, \cite{LL16} showed that multi-level lattices constructed from nested polar codes via integer partitions $\mathbb{Z}/2\mathbb{Z}/\cdots$ are capacity achieving with essentially the same decoding and encoding complexity as the AWGN channel. In a recent contribution \cite{BBG16}, the outage probability of two-dimensional multi-level lattices based on Reed-Muller codes was numerically evaluated. 
	In a related discrete-model, a novel hierarchic polar coding scheme using two phases of polarization was proposed in \cite{SOS} for binary-input block fading channels. Based on this scheme, a multi-level structure was then proposed for fading channels with additive exponential noise, using a noise decomposition technique.
		\section{Notation and Initial Definitions}
	\label{sec:notation}
	Following \cite{ViterboOggier}, we define the diversity of a set as follows.
	\begin{definition}For two vectors $\xx, \yy \in \mathbb{R}^n$ define  $\mathcal{I}_{\xx,\yy} = \left\{ i \in \left\{1,\ldots, n\right\} : x_i \neq y_i \right\}$. The \textit{diversity order} of a set $\Lambda \subset \mathbb{R}^n$ is given by $l =\displaystyle \min_{{\xx \neq \yy} {\xx,\yy\in\Lambda}} |\mathcal{I}_{\xx,\yy}|.$When $l = n$, we say that $\Lambda$ has \textit{full-diversity}.
	\end{definition}
	Let $\Lambda$ be a lattice in $\mathbb{R}^n$. If $\Lambda$ has full diversity, it has an associated positive \textit{product-distance} defined by $d_{\text{prod}}(\Lambda) = \min_{\xx \in \Lambda \nozero} \prod_{i=1}^n |x_i|.$	
	A \textit{fundamental region} $\mathcal{R}(\Lambda)$ for $\Lambda$, is any region whose translates by vectors of $\Lambda$ tile $\mathbb{R}^n$ . For any $\yy, \xx \in \mathbb{R}^n$ we say that $\yy = \xx \pmod \Lambda$ iff $\yy - \xx \in \Lambda$. By convention, we fix a fundamental region and denote by $\yy \pmod \Lambda$ the unique representative $\xx \in \mathcal{R}(\Lambda)$ such that $\yy = \xx \pmod \Lambda$. A \textit{partition chain} $\Lambda_1/\Lambda_2/\Lambda_3/\ldots/\Lambda_m$ is a lattice sequence such that $\Lambda_m \subset \Lambda_{m-1} \subset \ldots \subset \Lambda_1$.
			The \textit{volume-to-noise} ratio (VNR) of an $n$-dimensional lattice $\Lambda$ with respect to $\sigma>0$ is defined as $\gamma_{\Lambda}(\sigma) = {V(\Lambda)^{2/n}}/{\sigma^2}.$

	We write $\MM = \diag{m_1,\ldots,m_n}$ for a diagonal matrix with diagonal elements $m_i$. We say that $\MM_1 \succ \MM_2$ if the difference $\MM_1 - \MM_2$ is positive definite. We denote the $n \times n$ identity matrix by $\II_n$ or simple $\II$ when there is no ambiguity.
	\subsection{Model Description and Fundamental Limits}
	In this paper we focus on infinite-dimensional constellations, for which the Poltyrev limit replaces the notion of capacity. Let $\HH$ be a diagonal matrix in $\mathbb{R}^{n \times n}$. For $\xx \in \Lambda$ the received signal is given by
	\begin{equation}
	\yy = \HH \xx + \ww,
	\end{equation} 
	where $\ww \sim \mathcal{N}(0,\sigma^2\II_n)$. The matrix $\HH$ is assumed to be constant in the whole transmission and known to the receiver. After $T$ transmissions, the channel equation can be written in matrix or vectorized forms, respectively:
	\begin{equation}\label{eq:block-fading}
	\underbrace{\mathbf{Y}}_{n\times T}=\underbrace{\mathbf{H}}_{n\times n}\underbrace{\mathbf{X}}_{n \times T}+\underbrace{\mathbf{W}}_{n\times T} \mbox{ or } 
	\underbrace{\mathbf{y}}_{nT \times 1}=\underbrace{\mathcal{H}}_{nT \times nT}\underbrace{\mathbf{x}}_{nT \times 1}+\underbrace{\mathbf{w}}_{nT\times1},
	\end{equation}
	where $\mathcal{H} = \mathbf{I}_{T} \otimes \mathbf{H}$. For a real number $\mathcal{D} \geq 1$, let
	\begin{equation*}
	\mathbb{H}_{\infty}(\mathcal{D}) = \left\{\HH \in \diag{\mathbb{R}^{n\times n}}: \det \HH^{\text{tr}}\HH = \mathcal{D} \right\}.
	\label{eq:compoundModel}
	\end{equation*}
	By convention, we will denote $\mathbb{H}_{\infty}(1)$ by $\mathbb{H}_{\infty}$. In the ``infinite constellation setting'', where any point of a lattice can be transmitted, the notion of normalized log-density replaces the usual rate. For the block-fading model with blocks of size $n$ and $T$ transmissions we define the log density of a lattice $\Lambda_T \subset \mathbb{R}^{nT}$ as $(1/T) \log V(\Lambda).$ The smallest possible log density achievable by a sequence of lattices constellation with vanishing probability of error is called the \textit{Poltyrev limit}.
	
	Consider a sequence of lattices $\Lambda_1,\ldots,\Lambda_T,\ldots$ of increasing dimension $nT$. We say that it achieves the Poltyrev limit if
	\begin{equation}
		\frac{1}{T} \log V(\Lambda_T) \to \frac{n}{2}\log\left(2\pi e \sigma^2 \mathcal{D}^{2/n}\right),
		\label{eq:PoltyrevLimit}
	\end{equation}
	and the $P_e(\Lambda_T,\HH) \to 0$ vanishes as $T \to \infty$.
	Furthermore, if this sequence has vanishing probability of error \textit{for any} $\HH \in \mathbb{H}_{\infty}(\mathcal{D})$ we say that it is \textit{universal}. The above discussion gives rise to the notion of goodness for the block-fading channel. 
	\begin{definition}\cite{CLB16} We say that a sequence of lattices $\Lambda_T$ of increasing dimension $nT$ is universally good for the block-fading channel if for any VNR $\gamma_{\Lambda_T}(\sigma) > \frac{2 \pi e}{\mathcal{D}^{2/n}}$ and all $\mathbf{H} \in \mathbb{H}_\infty(\mathcal{D})$, $ P_e(\Lambda_T,\mathbf{H}) \to 0$.
		\label{def:goodLattices}
	\end{definition}
	The objective of the rest of the paper is to construct lattices which are good for the block-fading channel from multilevel constructions.
	\section{One-Level Universal Lattices}
	\label{sec:twoLevel}
	Known schemes of universally good lattices use one-level lattices based on random codes with increasing alphabet size. In order to build partitions with more levels for the compound channel, we start by revisiting the one-level construction \textit{\`a la} Forney et al. \cite{Forney00}. This construction has the advantage of being explicit and efficiently (two-stage) decodable, as compared to random lattices. 
	
	Let $\Lambda/\Lambda^\prime$ be a lattice partition with quotient $|\Lambda/\Lambda^\prime|=p$, for a prime or prime-power $p$.  We construct the lattice 
	\begin{equation} \mathcal{L} = (\Lambda^{\prime})^T + \mathcal{C},
	\label{eq:L}
	\end{equation}
	where $\mathcal{C}$ is a code in $(\Lambda/\Lambda^\prime)^T$. For example, in applications to the AWGN channel, $\Lambda$ is usually the one-dimensional lattice and $\Lambda^\prime = p\mathbb{Z}$, in which case the quotient $\mathbb{Z}/ p\mathbb{Z} \simeq \mathbb{F}_p$ and we can take the code $\mathcal{C}$ to be a $p$-ary code. In order to analyze the goodness of the construction \eqref{eq:L} in the block-fading case we will define the mod-$\Lambda$ block-fading channel and the $\HH\Lambda/\HH\Lambda^\prime$ compound channel. To summarize we will show that the following ``recipe'' produces a lattice $\mathcal{L}$ (Eq. \eqref{eq:L}) which is universally good:
	
		\begin{itemize}
		\item Choose $\Lambda/\Lambda^\prime$ to be a \textit{full}-diversity partition chain of dimension $n$, the size of each block $\HH$.
		\item Choose $\mathcal{C}$ to be an universally capacity achieving code for the set of $\HH\Lambda/\HH\Lambda^\prime$-channels.
	\end{itemize}

	\subsection{The mod-$\Lambda$ block-fading channel}
	\begin{figure}[!htb]
		\centering
		\includegraphics[scale=0.3]{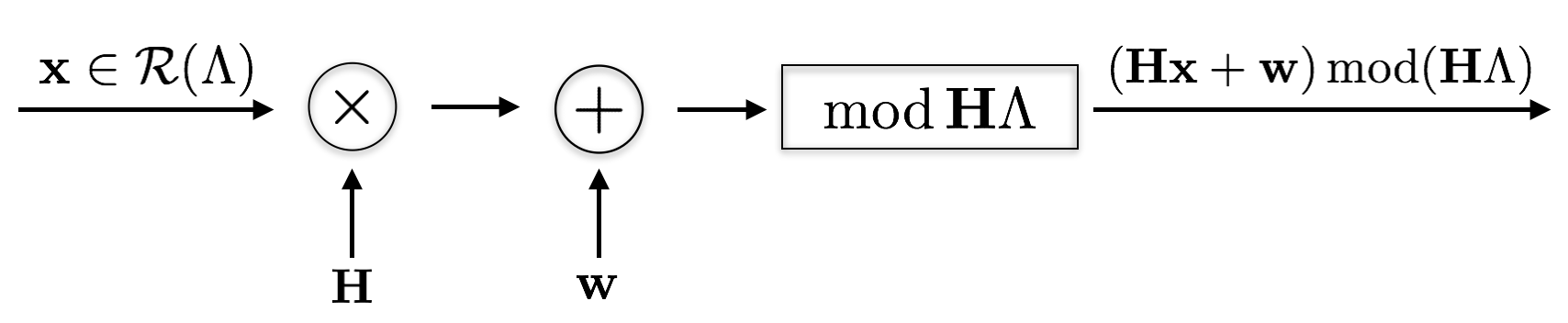}
		\caption{Diagram of the mod $\HH\Lambda$ channel}
	\end{figure}
	
	Given a vector $\xx \in \mathcal{R}(\Lambda)$ the block-fading mod-$\Lambda$ channel is a block-fading channel with the operation modulo $\HH \Lambda$ performed at the receiver front-end, namely
	$$\yy = (\HH \xx +\mathbf{w}) \pmod {\HH \Lambda}.$$
	where $\mathbf{H}$ is diagonal and belongs to $\mathbb{H}_{\infty}$, and $\ww \sim \mathcal{N}(0,\sigma^2 \II_n)$ as usual. Notice that the dimension of the lattices $\Lambda$ and $\Lambda^\prime$ is equal to block-size (number of rows of $\HH$).
	
	The $\HH \Lambda/ \HH \Lambda'$-channel is a modulo $\HH \Lambda^\prime$ channel with inputs constrained to be in $\Lambda$. Denote by $C_{\HH}(\Lambda,\sigma^2)$ the capacity of the mod-$\HH \Lambda$ channel and by $C_{\HH}(\Lambda/\Lambda^\prime,\sigma^2)$ the capacity of the $\HH \Lambda/\HH \Lambda^\prime$ channel, given a realization $\HH$. When $\HH = \mathbf{I}$ we use the shorthand notation $C_{\II}(\Lambda/\Lambda^\prime,\sigma^2) =C(\Lambda/\Lambda^\prime,\sigma^2)$. 
	
	We next adapt some results from multi-level partition to our block-fading model. A detailed analysis and proof of these results can be found in \cite{Forney00}, especially sections III and VI.
	First of all, we have the separability property
	\begin{equation} C_{\HH}(\Lambda/\Lambda^\prime,\sigma^2) = C_{\HH}(\Lambda^\prime,\sigma^2) - C_{\HH}(\Lambda,\sigma^2).
	\label{eq:conservationRule}
	\end{equation}
	For a one-level partition $\Lambda/\Lambda'$, where $\Lambda$ and $\Lambda'$ have decreasing and increasing volume, respectively, the capacity of the $\HH \Lambda/\HH \Lambda'$ channel tends to the Poltyrev limit, i.e.:
	\begin{equation}
		\lim_{{V(\Lambda') \to \infty} \above 0pt {V(\Lambda) \to 0} } \frac{C_{\HH}(\Lambda',\sigma^2)}{(n/2)\log(\gamma_{\Lambda^\prime}(\sigma)/2\pi e )} = 1.
		\label{eq:limitPoltyrev}
	\end{equation}
	For any \textit{fixed} $\HH$, the $\HH\Lambda/\HH\Lambda^\prime$ is \textit{regular}, i.e., the uniform input distribution is capacity achieving. The set of all $\HH \Lambda/\HH \Lambda^\prime$ channels, where $\HH \in \mathbb{H}_{\infty}(\mathcal{D})$, is therefore a \textit{compound channel} with capacity.
	\begin{equation} C_{\mathbb{H}_{\infty}}(\mathcal{D})(\Lambda/\Lambda^\prime,\sigma^2) = \inf_{\HH \in \mathbb{H}_\infty(\mathcal{D})} C_{\HH}(\Lambda/\Lambda^\prime,\sigma^2).
	\label{eq:capacityCompoundLevel}
	\end{equation}
	\subsection{Ideal lattices}
	Let $\mathcal{C}$ be a $p$-ary code for the $\HH\Lambda/\HH\Lambda^{\prime}$ channel and use Construction A. Upon receiving vector $\mathbf{y}=[\mathbf{y}_1, \mathbf{y}_2, \ldots, \mathbf{y}_T]$ where $\mathbf{y}_t \in \mathbb{R}^n$, the two-stage decoding rules consists of two stages:
	
	\begin{enumerate}
		\item Decode $\bar{\mathbf{y}}$, $\bar{\mathbf{y}}_i = \mathbf{y}_i \text{ mod } \mathbf{H}\Lambda^{\prime}$, to code $\mathcal{C}$, outputting $\mathbf{c}$.
		\item Decode $\bar{\mathbf{y}} - \mathbf{c}$ to $\HH\Lambda^\prime$
	\end{enumerate}
	
	In the second stage, we require $P_e(\Lambda^{\prime},\mathbf{H}) \approx 0$ \textit{for all} $\mathbf{H}$ with fixed determinant $\mathcal{D}$. This can be satisfied by a bounded minimum Euclidean distance decoder, provided that $\Lambda^{\prime}$ has full diversity.
	\begin{equation}
	\begin{split}
	& \min_{\mathbf{H}: |\mathbf{H}|=\mathcal{D}}d_{\min}^2(\mathbf{H}\Lambda^{\prime}) = \min_{\mathbf{H}: |\mathbf{H}|=\mathcal{D}}\min_{\mathbf{x}\in \Lambda^{\prime},\mathbf{x}\neq\mathbf{0}}\|\mathbf{Hx}\|^2 \\ &=\min_{\mathbf{x}\in \Lambda^{\prime},\mathbf{x}\neq\mathbf{0}} n \mathcal{D}^{2/n} (x_1 \cdots x_n)^{2/n} = n \mathcal{D}^{2/n} {d}_{\text{prod}}(\Lambda^{\prime})^{2/n}.
	\end{split}
	\end{equation}
	The second step follows from the inequality of arithmetic and geometric means; amazingly, the minimum distance is precisely known here, i.e., there exists a realization  $\mathbf{H}$ such that the equality holds (this realization is $h_i^2=\frac{(\mathcal{D}d_{prod}(\Lambda^{\prime}))^{2/n}}{x_i^2}$ where $x_1 \cdots x_n=d_{prod}(\Lambda^{\prime})$; see e.g. \cite[Prop. 3.6]{DBLP:journals/corr/LuzziV15}). Therefore, as long as the norm $d_{\text{prod}}$ is sufficiently large, reliable decoding of $\mathbf{H}\Lambda^{\prime}$ is possible for all $\mathbf{H}$. This requires the property that $\Lambda^\prime$ has \textit{full-diversity}. More precisely, the worst-case probability can be bounded as follows:
	
	\begin{proposition} With notation as above, as $d_{prod}(\Lambda^{\prime}) \to \infty$,
		\begin{equation}
		\begin{split}
		\sup_{\HH \in \mathbb{H}_{\infty}(\mathcal{D})} P_{e}(\HH\Lambda^{\prime}) \leq &P\left(\left\|\ww \right\| \geq \frac{(\sqrt{n} \mathcal{D} d_{\text{prod}}(\Lambda^{\prime}))^{\frac{1}{n}}}{2}\right) \to 0.\\
		\end{split}
		\label{eq:universalIdeal}
		\end{equation}
		\label{prop:1}
	\end{proposition}
	The right-hand side of Eq. \eqref{eq:universalIdeal} can be further bounded using $Q$-function bounds in order to estimate the product-distance for a target probability of error (see \cite{BB16} for more details).
	Proposition \ref{prop:1} strictly requires $\Lambda^\prime$ to have full diversity. This can be seen geometrically, for instance, if $\HH \in \mathbb{H}_{\infty}$ has some component arbitrarily close to zero. An explicit construction of full-diverse lattices is given by the generalized Construction A \cite{CLB16}, $\Lambda^{\prime}$ is the embedding in $\mathbb{R}^n$ of an ideal $\mathfrak{p}$ of the ring of integers of a number field $K$.

	In the first stage, the rate of the code $\mathcal{C}$ is bounded by the capacity $C_{\mathbf{H}}(\Lambda/\Lambda^{\prime}, \sigma^2)$ of the $\mathbf{H}\Lambda/\mathbf{H}\Lambda^{\prime}$ channel. In general, we do not know whether $C_{\mathbf{H}}(\Lambda/\Lambda^{\prime}, \sigma^2)$ is fixed, since $\mathbf{H}$ is arbitrary. However if $\Lambda$ has full diversity and we scale the bottom lattice $\Lambda$ appropriately we can make $C_{\mathbf{H}}(\Lambda,\sigma^2) \approx 0$ for any channel realization, so that in the limit \eqref{eq:limitPoltyrev}:
	\begin{equation}
	\begin{split}
	C_{\mathbf{H}}(\Lambda/\Lambda^{\prime}, \sigma^2) \approx 
\frac{n}{2}\log\left(\frac{(\mathcal{D}V(\Lambda^{\prime})^{2/n}}{2\pi e \sigma^2}\right)
	\label{eq:finalCapacity}
	\end{split}
	\end{equation}
	which only depends on $\mathcal{D}$. So it is possible to design a universal good lattice, provided that the code for the $\HH\Lambda/\HH\Lambda^{\prime}$ is universal. Examples include spatially-coupled LDPC codes or universal variants of polar codes \cite{KRU,SW,HU}.
	
	\section{Construction D}
	\label{sec:consD}
	In order to reduce the alphabet-size and improve decoding complexity, multilevel constructions are employed. We will then apply this construction in Section \ref{sec:numerical} in order to exhibit multi-level lattices based on nested polar codes as in \cite{PolarArXiv} which are good for block-fading channel. More specifically, we will employ multilevel partitions of \textit{algebraic} lattices as in \cite{GBB97}. Due to space limitations, we refer to \cite{ViterboOggier} for undefined algebraic terms.
	Assume a partition chain $\Lambda_0=\psi(\mathcal{O}_K), \Lambda_1=\psi(\mathfrak{p}_1), \ldots, \Lambda_m=\psi(\mathfrak{p}_r)$, where the $\mathfrak{p}_i$'s are nested ideals of a number field $\mathcal{O}_K$ of degree $n$. For a fixed $\HH$, iterating the capacity conservation rule \eqref{eq:conservationRule}, we obtain:
	\begin{equation}
	\sum_{i=1}^m C_{\mathbf{H}}(\Lambda_i/\Lambda_{i+1}, \sigma^2) = C_{\mathbf{H}}(\Lambda_0/\Lambda_{m}, \sigma^2),
	\end{equation}
	which is expected to be approximately the normalized log-density of the lattice. However we face two main problems
	
	\begin{enumerate} 
		\item Rate allocation: what is the maximum rate $R_i = C_{\mathbf{H}}(\Lambda_i/\Lambda_{i+1}, \sigma^2)$ for each level?
		\item Degradation: As we move along a lattice partition chain, the capacity of each level increases. However, in order to apply lattice construction from nested capacity-achieving polar codes of Section \ref{sec:numerical}, we need the stronger fact that such channels are stochastically degraded. Under which conditions can we guarantee degradation?
	\end{enumerate}
	The following useful result is adapted from \cite[Lem. 3]{PolarArXiv}:
	
	\begin{lemma}[Channel Degradation] If $|\DD| \succ \II$ is a diagonal matrix, the mod-$\Lambda$ channel is stochastically degraded with respect to the mod-$\DD\Lambda$ channel. Furthermore, if $\Lambda/\DD\Lambda/\DD^2\Lambda$ is a lattice partition chain, the $\Lambda/\DD\Lambda$-channel is stochastically degraded with respect to the $\DD\Lambda/\DD^2\Lambda$-channel.
		\label{lem:degradation}
	\end{lemma}
	
	\subsection{Rate Allocation}
	
	\begin{figure}[!htb]
		\centering
		\includegraphics[scale=0.30]{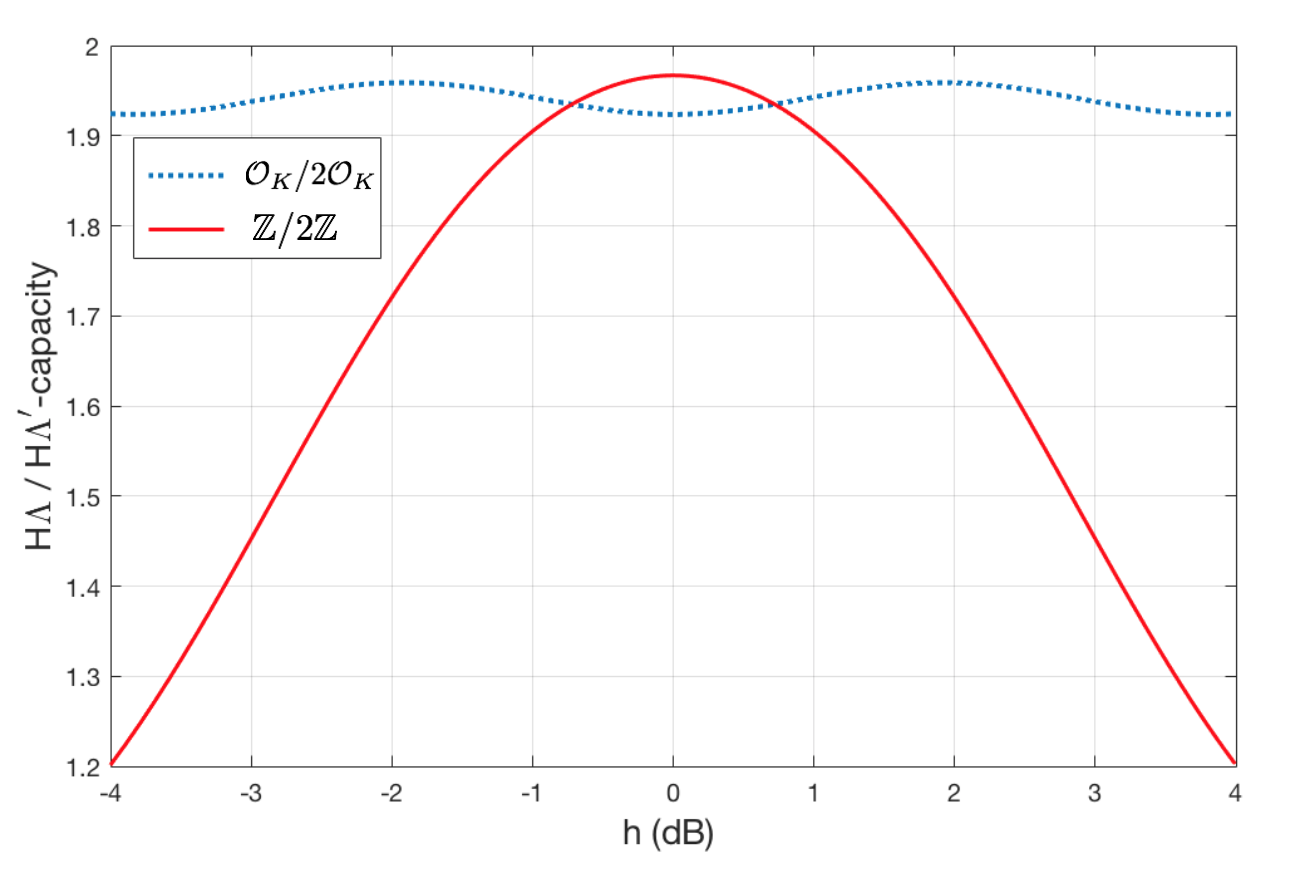}
		\caption{Capacity of the $\Lambda/\Lambda^{\prime}$ block-fading channel, as a function of the channel realization $\HH = \diag {h,1/h}$ for $(1/\sigma^2) = 10.5\text{dB}$.}
		\label{fig:OK2OK}
	\end{figure}
	For each $i$, the set of all symmetric $\HH \Lambda_{i}/\HH \Lambda_{i+1}$-channels, $\HH \in \mathbb{H}_{\infty}$, forms a compound channel with capacity given by \eqref{eq:capacityCompoundLevel}. Although we expect the sum of the capacities of the relevant levels to be equal to the normalized density \eqref{eq:finalCapacity}, the capacity of each level may depend on $\HH$. Partitions without additional algebraic structure may be very sensible to the variation of $\HH$, as illustrated in Fig. \ref{fig:OK2OK}. However, we show next that for algebraic partitions there are natural algebraic bounds for the compound capacity of each level. The bounds are of a theoretical interest. In practice, for low-degree ideal lattices we can we can identify precisely the best and worst channels in each compound set, by a careful algebraic analysis.
	
	\begin{proposition}
		Let $\Lambda$ be an ideal lattice of a totally real number field $K$ of degree $n$. For any $\HH \in \mathbb{H}_{\infty}$, the capacity of the mod-$\Lambda$ compound block fading channel is bounded as
		\begin{equation}
		C(\Lambda,(e^{\rho} \sigma)^2)\leq C_{\HH}(\Lambda,\sigma^2) \leq C(\Lambda,(e^{-\rho} \sigma)^2),
		\label{eq:boundCapacity}
		\end{equation}
		where $\rho$ is the covering radius of the log-unit lattice of $K$.
	\end{proposition}
	
	\begin{proof}
		From \cite[Lem. 4]{CLB16}, there exists a diagonal matrix $\EE$ such that 
			$|E_{ii}| \in \left[e^{-\rho}, e^{\rho} \right]$ and $C_{\HH}(\Lambda,\sigma^2) = C_{\EE}(\Lambda,\sigma^2)$.
		We claim that the capacity of the mod-$\EE\Lambda$ channel is nested between the capacity of the mod-$(e^{\pm \rho})\Lambda$ channels. For instance, for the upper bound, by scaling the channels, we only need to prove that the mod-$(e^{-\rho} \EE \Lambda)$ is less capable than the mod-$\Lambda$ channel. Since $\II \succ e^{-\rho} \EE $, this follows from Lem. \ref{lem:degradation}.
	\end{proof}
	
	\subsection{Degradation}
    We show how to construct chains that induce degraded channels using number fields.	Let $K/\mathbb{Q}$ be a Number Field and $p$ a rational prime. We consider two possibilities: 1) $p$ remains inert, 2) $p$ is totally ramified. In case 1), we have $p\mathcal{O}_K = \mathfrak{p}$, where $\mathfrak{p}$ is a prime ideal, and $\mathcal{O}_K \simeq \mathbb{F}_{p^n}$, and we obtain the algebraic partition $\mathcal{O}_K/p\mathcal{O}_K/p^2\mathcal{O}_K/\ldots$. Passing to $\mathbb{R}^n$ via the embeddings, we obtain $\mathbf{H} \psi(\mathcal{O}_K) / p \mathbf{H} \psi(\mathcal{O}_K) / p^2\mathbf{H} \psi(\mathcal{O}_K)/\ldots$. This partition is now self-similar and degradation follows from \cite[Lem. 3]{PolarArXiv}. The drawback is that we have to work with codes over $\mathbb{F}_{p^n}$, thus depending on the channel dimension the finite field will be very large. In case $2)$, $p\mathcal{O}_K = \mathfrak{p}^n$ and $\mathcal{O}_K/\mathfrak{p} \simeq \mathbb{F}_{p}$. The algebraic partition chain is then $\mathcal{O}_K/\mathfrak{p}/\mathfrak{p}^2/\ldots/\mathfrak{p}^{n-1}/p \mathcal{O}_K$. If, in addition $\mathfrak{p}$ is a principal ideal, then each $\mathfrak{p}^i = \left\langle \alpha^i\right\rangle$ where $\alpha \in \mathcal{O}_K$ is such that $\left\langle \alpha^n\right\rangle = p \mathcal{O}_K$. Therefore, passing to the real case, we have $\mathbf{H}  \psi(\mathcal{O}_K)/ \mathbf{H} \mathbf{D}_{\alpha} \psi(\mathcal{O}_K) / \mathbf{H}  \mathbf{D}_{\alpha}^2 \psi(\mathcal{O}_K)$, where $D_{\alpha}$ is a diagonal matrix with $[D_{\alpha}]_{ii} = \sigma_i(\alpha)$. Building ideals to ensure $|[D_{\alpha}]_{ii}| \succ I$, we can apply Lem. \ref{lem:degradation}.
	
	\subsection{Worked Out Examples}
	\textit{Binary Partitions.}
	Let $\mathbb{Z}[\sqrt{2}]= \left\{ a+b\sqrt{2} : a,b \in \mathbb{Z}\right\}$ be the ring of integers of the field extension $\mathbb{Q}(\sqrt{2})$. A simple example of binary partition is given by $\mathbb{Z}[\sqrt{2}]/\sqrt{2}\mathbb{Z}[\sqrt{2}]/2 \mathbb{Z}[\sqrt{2}]/\cdots$ and its $2^n$-dimensional generalization, which have been thoroughly analyzed in \cite{GBB97}, in terms of coding gain and product norm. Here we expand the analysis to the capacity of $\mod\HH\Lambda$ channels. More precisely, let 
	$$\psi(a+b\sqrt{2})=(a+b\sqrt{2},a-b\sqrt{2}),$$
	be the embedding of an algebraic number of the form $a+b\sqrt{2}$. Then we can construct the two-dimensional binary partition $\psi(\mathbb{Z}[\sqrt{2}])/ \psi(\sqrt{2} \mathbb{Z}[\sqrt{2}])/ \psi(2\mathbb{Z}[\sqrt{2}])/\ldots$. We can see from Lem. \ref{lem:degradation}, that the channels in each level are degraded with respect to the channel in the next level. For the rate allocation, we identify the worst channel. Let 
	$$C_{\HH}(\Lambda/\Lambda^\prime,\sigma) = h(\HH \Lambda, \sigma^2) - h(\HH\Lambda^\prime , \sigma^2) + 1$$
	be the capacity of the first level, where $\Lambda = \psi(\mathbb{Z}[\sqrt{2}] )$ and $\Lambda^\prime =\psi(\sqrt{2} \mathbb{Z}[\sqrt{2}] )$. For any $\HH \in \mathbb{H}_{\infty}$ denote by $C_{h}(\Lambda/\Lambda^\prime,\sigma)$ the capacity of the mod-$\Lambda$ fading channel for fading realization $(h,1/h)$. Due to the invariance of the partition by units in $\mathbb{Z}[\sqrt{2}]$, $C_{h}(\Lambda/\Lambda^\prime,\sigma)$ is a multiplicatively periodic function, i.e. $C_{h}(\Lambda/\Lambda^\prime,\sigma) = C_{(1\pm\sqrt{2})h}(\Lambda/\Lambda^\prime,\sigma), \mbox{ for any } h > 0.$
	From this, and from the fact that $C_{h}(\Lambda/\Lambda^\prime,\sigma) = C_{1/h}(\Lambda/\Lambda^\prime,\sigma)$, we can conclude that $h = 1$ and $h = \sqrt{1+\sqrt{2}}$ are extreme points of the capacity (see Fig. \ref{fig:capacity_periodic}). Interestingly, \cite{GBB97} have shown that these are also the extreme points of Hermite parameter of $\HH \Lambda$. The rate should be allocated according to the worst channel. We also observe a ``phase-transition'' in the worst-case channel in Fig. \ref{fig:capacity_periodic}: for low levels $h=\sqrt{1+\sqrt{2}}$ is the worst-case while $h=1$ becomes the best case as the level increases (the precise transition point is when $(\vol \,\, \Lambda_i)^{1/n} = 2 \pi e \sigma^2$, where $i$ represents the $i$-th level of the partition). The sums of the capacities of all levels is approximately the Poltyrev limit.
	
	\begin{figure}[!htb]
		\centering
		\includegraphics[scale=0.30]{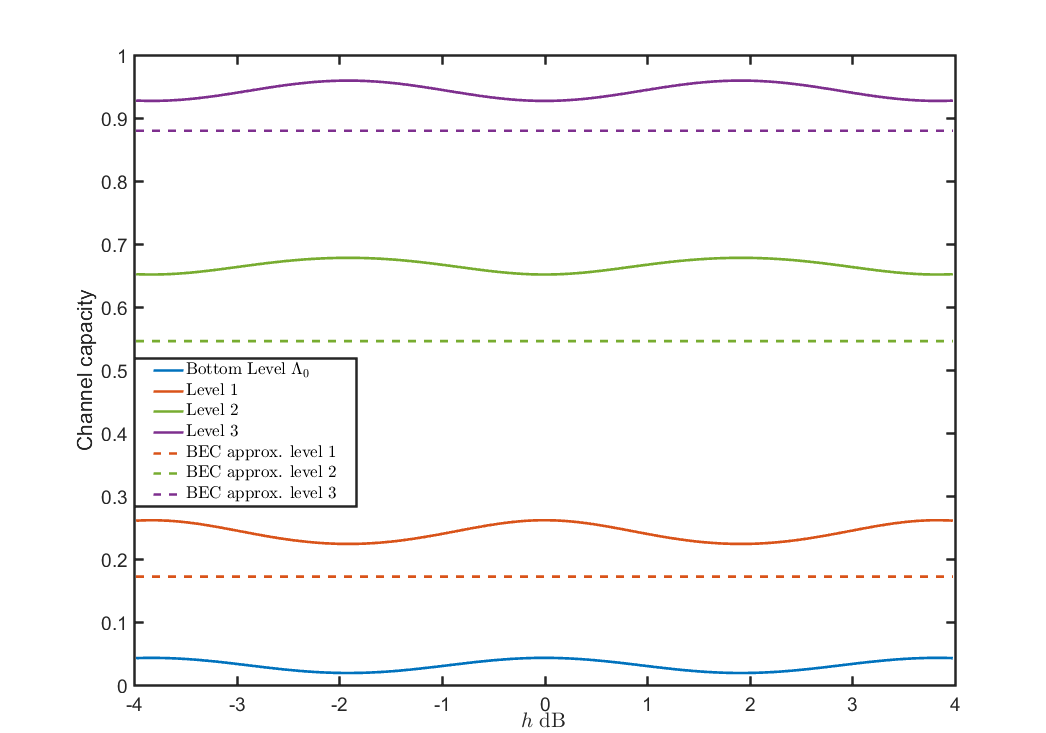}
		\caption{Periodic capacity functions of the partition channels. Note that to make $C_h(\mathbb{Z}[\sqrt{2}],\sigma)$ negligble, the partition chain is scaled with $\eta=0.5$.}
		\label{fig:capacity_periodic}
	\end{figure}

	\section{Numerical Results}
			\label{sec:numerical}
			\begin{figure}[!htb]
				\includegraphics[scale=0.15]{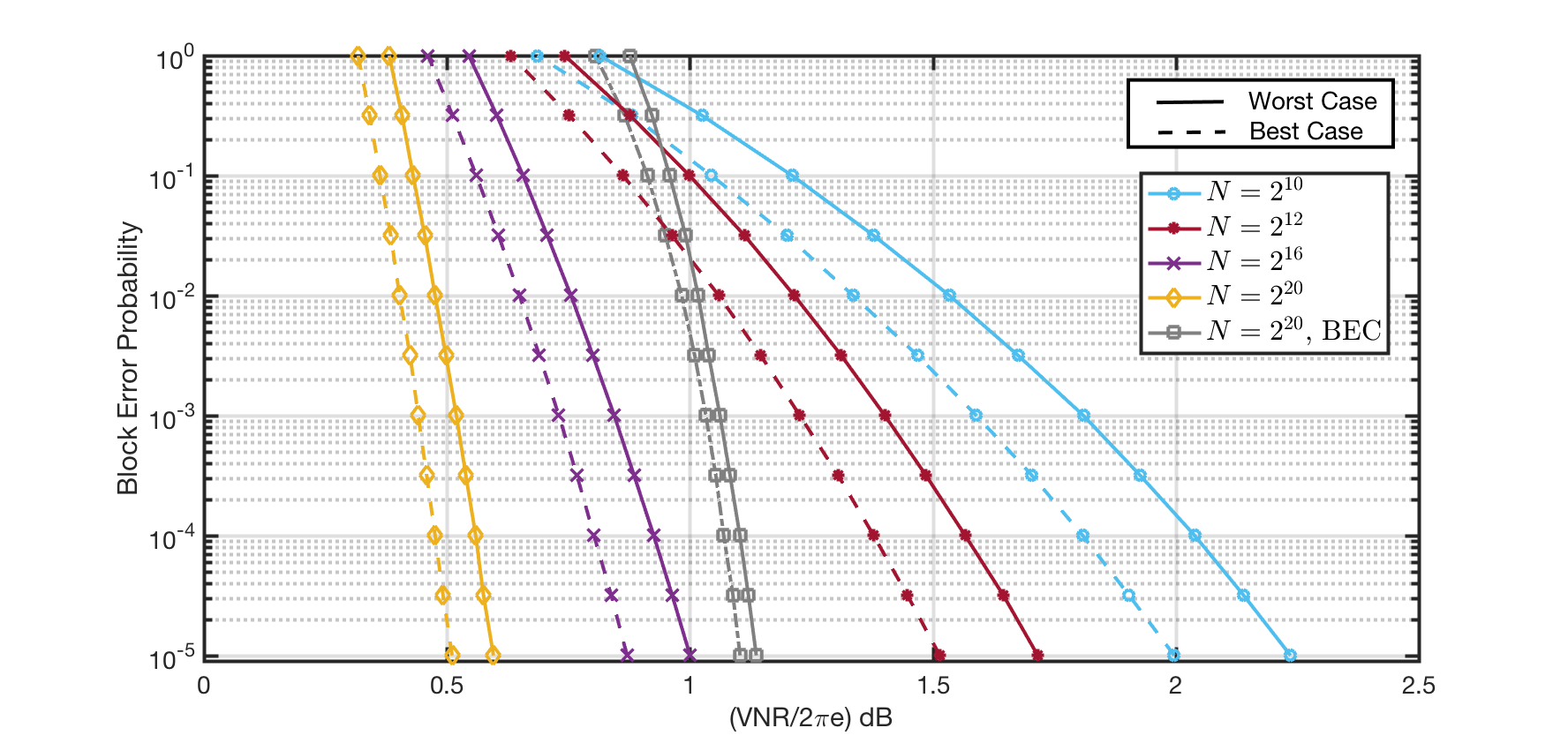}
				\caption{Performance of universal polar lattices in compound fading channels.}
				\label{fig:VNRperform}
			\end{figure}
			
	In this section, we design explicit universal polar lattices for compound fading channels, using universal polar codes \cite{HU} and construction D. Let $\Lambda = \mathbb{Z}[\sqrt{2}]$. To simplify the design, we choose the binary partition chain $\Lambda/\sqrt{2}\Lambda/2\Lambda/\ldots$ and binary universal polar codes. The compound fading channel is also assumed to have only two extreme cases with $h_1=1$ and $h_2=\sqrt{1+\sqrt{2}}$. As can be seen from Fig. \ref{fig:capacity_periodic}, the capacity $C_h(\mathbb{Z}[\sqrt{2}],\sigma)$ is negligible. Our aim is to design a universal polar code for each partition channel, with rate no more than $C_{h_2}(\Lambda/\sqrt{2}\Lambda,\sigma)=0.2239$, $C_{h_1}(\sqrt{2}\Lambda/2\Lambda,\sigma)=0.6516$, and $C_{h_1}(2\Lambda/2\sqrt{2}\Lambda,\sigma)=0.9271$, respectively. Since there are two fading states, we can employ the chaining construction \cite[Sec. V]{HU} to build a universal polar code for each level. Moreover, by Lem. \ref{lem:degradation} and the construction of polar lattices \cite{PolarArXiv}, we indeed generate a new polar lattice by chaining the two standard polar lattices corresponding to $h_1$ and $h_2$, respectively. The performance of the constructed codes is shown in Fig. \ref{fig:VNRperform}. The gap between worst and best cases can be as small as $0.1$dB, while the distance to the theoretical limit is kept as $0.6$dB, for a target probability of error of $10^{-5}$. Although strictly speaking the codes are provably guaranteed to work for the two extremes, we observe in simulations a negligible error probability for the channels in between.  This scheme can be easily generalized to compound channels with finite states (e.g. \cite[Sec. V-C]{HU}). 
	
	In the case of infinite states, we can use the universal polarization technique proposed in \cite{SW} to construct a universal polar code for each level. However, the polarization rate is slowed down due to the slow polarization stage. A more convenient construction is to use binary erasure channel approximation, i.e., to find a upper bound $\epsilon$ on the Bhattacharyya parameter of a partition channel, and construct a polar code for a surrogate BEC with erasure probability $\epsilon$. This polar code is universal (for \textit{all} channels within the set) because among all channels with a fixed Bhattacharyya parameter, the BEC's polarized descendants have the highest Bhattacharyya parameters. Unfortunately, we have to sacrifice some data rate by using this method. The maximum achievable rate by BEC approximation for each level is also shown in Fig. \ref{fig:capacity_periodic}, indicating that the rate loss is still acceptable. In this regime, we observe a loss of approximately $0.5$dB compared to the previous scheme. 
	\section{Conclusion}
	We have considered practical constructions and guidelines for the construction of universal multi-level lattice codes. Both theoretical bounds and numerical results were exhibited. Concretely, we illustrated the constructions with lattices from polar codes for the $2$-dimensional compound-fading channel, showing that a small gap to the Poltyrev limit can be universally achieved in all channels. On a qualitative level, our analyses corroborate the fact that algebraic partitions can be greatly advantageous as compared to unstructured ones, even from an information-theoretic perspective.
	
	\section*{Acknowledgments}
	The work of L. Liu is supported by Huawei's Shield Lab through the HIRP Flagship Program. The authors would like to thank J.-C. Belfiore for fruitful discussions.	
	\bibliographystyle{plain}
	\bibliography{block_fading}

\end{document}